\documentclass[12pt,a4paper]{article}%
\usepackage{amssymb}
\usepackage{amsfonts}
\usepackage{amsmath}
\usepackage{graphicx}
\usepackage[singlespacing]{setspace}
\usepackage{geometry}
\usepackage{url}
\usepackage{hyperref}%
\providecommand{\U}[1]{\protect\rule{.1in}{.1in}}
\newtheorem{theorem}{Theorem}

\newtheorem{proposition}[theorem]{Proposition}

\newenvironment{proof}[1][Proof]{\noindent\textbf{#1.} }{\ \rule{0.5em}{0.5em}}
\ifx\pdfoutput\relax\let\pdfoutput=\undefined\fi
\newcount\msipdfoutput
\ifx\pdfoutput\undefined\else
\ifcase\pdfoutput\else
\ifx\paperwidth\undefined\else
\ifdim\paperheight=0pt\relax\else\pdfpageheight\paperheight\fi
\ifdim\paperwidth=0pt\relax\else\pdfpagewidth\paperwidth\fi
\fi\fi\fi
\begin{document}

\title{\textbf{Imperfect Credibility \textit{versus} No Credibility of Optimal
Monetary Policy}\thanks{We thank two anonymous referees for very useful
comments, Catherine Refait-Alexandre, Laurent Weill and participants to our
session at the 36th GDRE symposium on Money Banking and Finance in
Besan\c{c}on, June 2019. PSE thanks the support of the EUR grant
ANR-17-EURE-001. }}
\author{Jean-Bernard Chatelain\thanks{Paris School of Economics, Universit\'{e} Paris
1 Pantheon Sorbonne, PjSE, 48 Boulevard Jourdan, 75014 Paris. Email:
jean-bernard.chatelain@univ-paris1.fr}\ \quad Kirsten Ralf\thanks{ESCE
International Business School INSEEC U Research Center, 10 rue Sextius Michel,
75015 Paris, Email: Kirsten.Ralf@esce.fr.}}
\maketitle

\begin{abstract}
A minimal central bank credibility, with a non-zero probability of not
renegning his commitment ("quasi-commitment"), is a necessary condition for
anchoring inflation expectations and stabilizing inflation dynamics. By
contrast, a complete lack of credibility, with the certainty that the policy
maker will renege his commitment ("optimal discretion"), leads to the local
instability of inflation dynamics. In the textbook example of the
new-Keynesian Phillips curve, the response of the policy instrument to
inflation gaps for optimal policy under quasi-commitment has an opposite sign
than in optimal discretion, which explains this bifurcation.

\textbf{JEL\ classification numbers}: C61, C62, E31, E52, E58.

\textbf{Keywords:} Ramsey optimal policy, Imperfect commitment, Discretion,
New-Keynesian Phillips curve, Imperfect credibility.

\textbf{POSTPRINT published in: }\textit{Revue Economique, }(2021), 72(1),
43-53.\bigskip

\textbf{Titre en fran\c{c}ais}: \textbf{Cr\'{e}dibilit\'{e} imparfaite ou
absence de cr\'{e}dibilit\'{e} de la politique mon\'{e}taire optimale.}

\textbf{R\'{e}sum\'{e}}:

Une cr\'{e}dibilit\'{e} minimale de la banque centrale, avec une
probabilit\'{e} non-nulle de ne pas revenir sur son engagement, est une
condition n\'{e}cessaire pour ancrer les anticipations d'inflation et garantir
le retour de l'inflation vers sa cible de long terme. En revanche, l'absence
compl\`{e}te de cr\'{e}dibilit\'{e}, avec la certitude que la banque centrale
reviendra sur son engagement ("discr\'{e}tion optimale"), implique une
bifurcation de la dynamique de l'inflation vers des trajectoires
d\'{e}flationnistes ou hyper-inflationnistes. Dans l'exemple de la courbe de
Phillips des nouveaux Keyn\'{e}siens, la politique optimale \`{a} tr\`{e}s
faible cr\'{e}dibilit\'{e} impose une r\'{e}ponse de l'instrument de politique
mon\'{e}taire aux \'{e}carts de l'inflation de signe oppos\'{e} \`{a} la
politique de discr\'{e}tion optimale, ce qui explique cette bifurcation.

\textbf{Mots-cl\'{e}s:} Politique optimale \`{a} la Ramsey, discr\'{e}tion,
cr\'{e}dibilit\'{e} imparfaite, engagement limit\'{e}, Courbe de Phillips des
nouveaux Keyn\'{e}siens.\newpage

\end{abstract}

\section{INTRODUCTION}

A key determinant of the efficiency of stabilization policy is the degree of
the credibility of policy makers, measured by their probability of not
reneging their commitment or the probability of a change of the head of the
central bank. This paper shows that the equilibrium in a model of optimal
policy under quasi-commitment (Schaumburg and Tambalotti [2007], Fujiwara, Kam
and\ Sunakawa [2019]) is completely different from the discretion equilibrium
with no commitment and more relevant for policy makers, even if the
probability of non-reneging tends to zero (near-zero credibility). This latter
discretion equilibrium, however, is presented as a relevant benchmark
equilibrium in numerous papers since Clarida, Gali and Gertler [1999], for
example in Gali [2015].

In the discretion equilibrium, each period-specific policy maker does a static
optimization ignoring expectations, whereas the monetary policy transmission
mechanism includes dynamics related to private sector expectations. Phillips,
already in 1954, warned about the dramatic errors of static optimization when
the underlying transmission mechanism is dynamic:

\begin{quotation}
The time path of income, production and employment during the process of
adjustment is not revealed. It is quite possible that certain types of policy
may give rise to undesired fluctuations, or even cause a previously stable
system to become unstable, although the final equilibrium position as shown by
a static analysis appears to be quite satisfactory (Phillips [1954], p.~290).
\end{quotation}

In discretion, the policy maker's rule, which is optimal for a static model,
leads to a sub-optimal positive feedback mechanism once expectations are taken
into account. It results in a locally unstable equilibrium in the space of
inflation and of the cost-push shock. Therefore, discretion equilibrium
requires with an infinite precision the knowledge of the parameters of the
monetary policy transmission in order to force an exact correlation between
inflation and the cost-push shock. In practice, this perfect knowledge never
occurs, so that the probability to shift to inflation or deflation spirals is
equal to one in the discretion equilibrium.

When changing the policy from commitment or quasi-commitment to discretion the
qualitative properties of the dynamic system change dramatically. The
probability of reneging commitment serves as a bifurcation parameter, the
bifurcation occurs when going from probability zero to any probability, even
\emph{infinitesimal} small, larger than zero. With quasi-commitment, the
private sector's expectations of inflation are taken into account in the
equilibrium, even with an extremely low probability of not reneging
commitment. The policy maker's negative feedback rule is such that the policy
instrument responds with an opposite sign to inflation as in discretion
equilibrium. This ensures the locally stable dynamics of inflation. The
inflation auto-correlation (or growth factor) parameter shifts from above one
(discretion) to below one (quasi-commitment). Shifting from discretion
equilibrium to quasi-commitment corresponds to a saddle-node bifurcation of
inflation dynamics.

However, the existing literature did not mention the bifurcation when shifting
from quasi-commitment to discretion equilibrium. Schaumburg and Tambalotti's
([2007], p.304) statement that "\emph{quasi-commitment converges to full
commitment for [the probability of reneging commitment tends to zero]}" is
valid. But their second statement that "\emph{it also converges to discretion
when [the probability of reneging commitment tends to one]}" is not valid, as
demonstrated in this paper. Schaumburg and Tambalotti's ([2007], figure 4,
p.318) numerical examples suggest that their second statement is likely to be
false. There is a gigantic gap between the initial jumps of inflation which is
nearly the \emph{double} for discretion ($4.9$) as compared to optimal policy
under quasi-commitment ($2.5$) for the lowest numerical value of the
probability of not reneging commitment that they have chosen. Because of this
large gap, Schaumburg and Tambalotti [2007] conclude that "\emph{most of gains
of commitment accrue at relatively low levels of credibility}". We show that
most of the large gains of commitment accrue at extremely low levels of
credibility. Finally, the impulse response functions, welfare losses and
initial anchors (or jumps) of inflation are much larger with discretion than
with near-zero credibility.

Section 2 presents Ramsey optimal policy under imperfect commitment and
section 3 discretionary policy. Section 4 demonstrates the existence of a
bifurcation and calculates policy rule parameters comparing near-zero
credibility \emph{versus} discretion. Section 5 compares initial anchors of
inflation on the cost-push shock, impulse response functions and welfare for
near-zero credibility \emph{versus} discretion. The last section concludes.

\section{RAMSEY\ OPTIMAL\ POLICY\ UNDER\ QUASI-COMMITMENT}

Following Schaumburg and Tambalotti [2007], we assume that the mandate to
minimize the loss function is delegated to a sequence of policy makers
(indexed by $j$, $k$,...) with a commitment of random duration. The length of
their tenure depends on a sequence of exogenous independently and identically
distributed Bernoulli signals $\left\{  \eta_{t}\right\}  _{t\geq0}$ with
$E_{t}\left[  \eta_{t}\right]  _{t\geq0}=1-q$, with $0<q\leq1$. The case $q=0$
of discretionary policy is treated separately in the next section. If
$\eta_{t}=1,$ a new policy maker takes office at the beginning of time $t$ and
is not committed to the policy of his/her predecessor. Otherwise, the
incumbent stays on. A higher probability of not reneging commitment $q$ can be
interpreted as a higher credibility. This leads to use a "credibility
adjusted" discount factor $\beta q$ in the policy maker's optimal behavior. A
policy maker with little credibility does not give a large weight on future
welfare losses.

At the start of his tenure, policy maker $j$ solves the following problem,
where subscript $k$ corresponds to the new policy maker. Welfare is maximized
subject to the new-Keynesian Phillips curve with slope $\kappa$ and subject to
the auto-correlation $\rho$ of a cost-push shock $u_{t}$ and the constant
variance $\sigma_{u}^{2}$ of its identically and independently normally
distributed disturbances $\eta_{u,t}$:%

\begin{align}
V^{jj}\left(  u_{0}\right)   &  =E_{0}%
{\displaystyle\sum\limits_{t=0}^{t=+\infty}}
\left(  \beta q\right)  ^{t}\left[  -\frac{1}{2}\left(  \pi_{t}^{2}%
+\frac{\kappa}{\varepsilon}x_{t}^{2}\right)  +\beta\left(  1-q\right)
V^{jk}\left(  u_{t}\right)  \right] \\
\text{s.t. }\pi_{t}  &  =\kappa x_{t}+\beta qE_{t}\pi_{t+1}+\beta\left(
1-q\right)  E_{t}\pi_{t+1}^{k}+u_{t}\text{ (Lagrange multiplier }\gamma
_{t+1}\text{)}\nonumber\\
u_{t}  &  =\rho u_{t-1}+\eta_{u,t}\text{, }\forall t\in%
\mathbb{N}
\text{, }u_{0}\text{ given, }0<\rho<1\text{, }\varepsilon>1\text{, }%
0<\beta<1\text{,}\nonumber
\end{align}

where $x_{t}$ represents the welfare-relevant output gap, i.e. the deviation
between (log) output and its efficient level, $\pi_{t}$ the rate of inflation
between periods $t-1$ and $t$, $\beta$ the discount factor, and $E_{t}$ the
expectation operator. The utility the central bank obtains if next period's
objectives change is denoted $V^{jk}$. Inflation expectations are an average
between two terms in the new-Keynesian Phillips curve. The first term, with
weight $q$ is the inflation that would prevail under the current regime upon
which there is commitment. The second term with weight $1-q$ is the inflation
that would be implemented under the alternative regime, which is taken as
given by the current central bank. The key change with respect to the model
without the possibility of a change of monetary policy is that the narrow
range of values for the discount factor around $0.99$ for quarterly data
($4\%$ discount rate) is much wider for the "credibility weighted discount
factor" of the policy maker: $\beta q\in\left]  0,0.99\right]  $, with limit
numerical value $q=10^{-7}>0$ in this paper.

The slope of the new-Keynesian Phillips is a decreasing function of the
household's elasticity of substitution between each differentiated
intermediate goods ($\varepsilon>1$):%

\begin{align}
\underset{\varepsilon\rightarrow+\infty}{\lim}\kappa &  =0<\kappa=\left(
\sigma+\frac{\varphi+\alpha_{L}}{1-\alpha_{L}}\right)  \frac{\left(
1-\theta\right)  \left(  1-\beta\theta\right)  }{\theta}\frac{\left(
1-\alpha_{L}\right)  }{\left(  1-\alpha_{L}+\alpha_{L}\varepsilon\right)
}<\kappa_{\max}=\underset{\varepsilon\rightarrow1^{+}}{\lim}\kappa\\
\text{with }\varepsilon &  >1\text{, }0<\beta,\alpha_{L},\theta<1\text{,
}\sigma>0\text{, }\varphi>0\text{.}\nonumber\\
\kappa_{\max}  &  =\underset{\varepsilon\rightarrow1^{+}}{\lim}\kappa=\left(
\sigma+\frac{\varphi+\alpha_{L}}{1-\alpha_{L}}\right)  \frac{\left(
1-\theta\right)  \left(  1-\beta\theta\right)  }{\theta}\left(  1-\alpha
_{L}\right)  .\nonumber
\end{align}

Gali's ([2015], chapter 3) calibration of structural parameters is as follows:
The representative household's discount factor $\beta=0.99$ for a logarithmic
utility of consumption $\sigma=1$ and a unitary Frisch elasticity of labor
supply $\varphi=1$. The household's elasticity of substitution between each
differentiated intermediate goods $\varepsilon=6$. The production function is
$Y=A_{t}L^{1-\alpha_{L}}$ where $Y$ is output, $L$ is labor, $A_{t}$
represents the level of technology. The measure of decreasing returns to scale
of labor is $0<\alpha_{L}=1/3<1$. The proportion of firms who do not reset
their price each period $0<\theta=2/3<1$ which corresponds to an average price
duration of three quarters. For these parameters, the maximal value of the
slope of the new-Keynesian Phillips curve when varying the elasticity of
substitution between intermediate goods, $\kappa_{\max}=0.34$ is obtained when
the elasticity of substitution tends to one. The auto-correlation of the
cost-push shock is $\rho=0.8$.

If the policy maker is maximizing welfare (Gali [2015]), the cost of changing
the policy instrument $x_{t}$ is scaled by $\frac{\kappa}{\varepsilon}$ which
is a decreasing function of the household's elasticity of substitution between
each differentiated intermediate goods ($\varepsilon>$1]:%

\[
0<\frac{\kappa}{\varepsilon}=\left(  \sigma+\frac{\varphi+\alpha_{L}}%
{1-\alpha_{L}}\right)  \frac{\left(  1-\theta\right)  \left(  1-\beta
\theta\right)  }{\theta}\frac{\left(  1-\alpha_{L}\right)  }{\left(
1-\alpha_{L}+\alpha_{L}\varepsilon\right)  }\frac{1}{\varepsilon}%
<\kappa\left(  \varepsilon\right)  <\kappa_{max}.
\]

With Gali's [2015] parameters, the relative weight of the variance of the
policy instrument (output gap) is a very low proportion ($\frac{\kappa
}{\varepsilon}=2.125\%$) of the weight on the variance of the policy target (inflation).

Differentiating the Lagrangian with respect to the policy instrument (output
gap $x_{t}$) and to the policy target (inflation $\pi_{t}$) yields the first
order conditions for $t=1,2,...$:%

\[
\left\{
\begin{array}
[c]{c}%
\frac{\partial%
\mathcal{L}%
}{\partial\pi_{t}}=0:\pi_{t}+\gamma_{t+1}-\gamma_{t}=0\\
\frac{\partial%
\mathcal{L}%
}{\partial x_{t}}=0:\frac{\kappa}{\varepsilon}x_{t}-\kappa\gamma_{t+1}=0
\end{array}
\right.  \Rightarrow\left\{
\begin{array}
[c]{c}%
x_{t}=x_{t-1}-\varepsilon\pi_{t}\\
x_{t}=\varepsilon\gamma_{t+1}=\varepsilon(\gamma_{t}-\pi_{t})
\end{array}
\right.
\]
The central bank's Euler equation ($\frac{\partial%
\mathcal{L}%
}{\partial\pi_{t}}=0$) links recursively the future or current value of
central bank's policy instrument $x_{t}$ to its current or past value
$x_{t-1}$, because of the central bank's relative cost of changing her policy
instrument is strictly positive $\alpha_{x}=\frac{\kappa}{\varepsilon}>0$.
This non-stationary Euler equation adds an unstable eigenvalue in the central
bank's Hamiltonian system including three laws of motion of one
forward-looking variable (inflation $\pi_{t}$) and of two predetermined
variables $\left(  u_{t},x_{t}\right)  $ or $\left(  u_{t},\gamma_{t}\right)
$.

The transversality condition $\gamma_{0}=0$ minimizes the loss function with
respect to inflation at the initial date:%

\[
\gamma_{0}=0\Rightarrow x_{-1}=-\varepsilon\gamma_{0}=0\text{ so that }\pi
_{0}=-\frac{1}{\varepsilon}x_{0}\text{ or }x_{0}=-\varepsilon\pi_{0}\text{ .}%
\]

It predetermines the policy instrument which allows to anchor the
forward-looking policy target (inflation). The inflation Euler equation
corresponding to period $0$ is not an effective constraint for the central
bank choosing its optimal plan in period $0$. The former commitment to the
value of the policy instrument of the previous period $x_{-1}$ is not an
effective constraint. The policy instrument is predetermined at the value zero
$x_{-1}=0$ at the period preceding the commitment. Combining the two first
order conditions to eliminate the Lagrange multipliers yields the optimal
initial anchor of forward inflation $\pi_{0}$ on the predetermined policy
instrument $x_{0}$.

Chatelain and\ Ralf's [2019] algorithm for Ramsey optimal policy with forcing
variables seeks a stable subspace of dimension two in a system of three
equations including the first order Euler condition on the policy instrument
(or on the Lagrange multiplier on inflation). The representation of the
optimal policy rule depends on current private sectors variables:%

\begin{equation}
x_{t}=F_{\pi}\pi_{t}+F_{u}u_{t}.
\end{equation}

This representation of the optimal policy rule is simpler than other
observationally equivalent alternatives proposed for example by Gali [2015] or
by Schaumburg and Tambalotti [2007], where the policy instrument depends on
its lagged value instead of inflation. Chatelain and Ralf [2019] provide the
details of the solution. The characteristic polynomial of the Hamiltonian
system is:%

\[
\lambda^{2}-\left(  1+\frac{1}{\beta q}+\frac{\kappa\varepsilon}{\beta
q}\right)  \lambda+\frac{1}{\beta q}=0.
\]

Its solution inside the unit circle is denoted "inflation eigenvalue"
$\lambda\left(  q,\varepsilon\right)  $:\qquad%

\[
\lambda\left(  q,\varepsilon\right)  =\frac{1}{2}\left(  \left(  1+\frac
{1}{\beta q}+\frac{\varepsilon\kappa}{\beta q}\right)  -\sqrt{\left(
1+\frac{1}{\beta q}+\frac{\varepsilon\kappa}{\beta q}\right)  ^{2}-\frac
{4}{\beta q}}\right)  =\frac{1}{\beta q}-\frac{\kappa}{\beta q}F_{\pi}.
\]

The optimal policy rule parameters are:%

\[
F_{\pi}=\left(  \frac{\lambda\left(  q,\varepsilon\right)  }{1-\lambda\left(
q,\varepsilon\right)  }\right)  \varepsilon>0\text{ and }F_{u}=\frac
{-1}{1-\beta q\rho\lambda\left(  q,\varepsilon\right)  }F_{\pi}.
\]

The initial value of the policy instrument $x_{0}$ is anchored on the initial
value of the cost-push shock $u_{0}$ because of the feedback policy rule. This
implies an optimal initial anchor of inflation $\pi_{0}$ on the initial value
of the cost-push shock $u_{0}$:%

\[
\left.
\begin{array}
[c]{c}%
x_{0}=F_{\pi}\pi_{0}+F_{u}u_{0}\\
\pi_{0}=-\frac{1}{\varepsilon}x_{0}%
\end{array}
\right\}  \Rightarrow\left\{
\begin{array}
[c]{c}%
x_{0}=\varepsilon\frac{-F_{u}}{F_{\pi}+\varepsilon}u_{0}=-\varepsilon
\frac{\lambda\left(  q,\varepsilon\right)  }{1-\beta q\rho\lambda\left(
q,\varepsilon\right)  }u_{0}\\
\pi_{0}=-\frac{F_{u}}{F_{\pi}+\varepsilon}u_{0}=\frac{\lambda\left(
q,\varepsilon\right)  }{1-\beta q\rho\lambda\left(  q,\varepsilon\right)
}u_{0}.
\end{array}
\right.  \text{ }%
\]

The dynamical system in the space of target variables is locally stable with
two eigenvalues inside the unit circle: $0<\rho<1$ and $0<\lambda\left(
q,\varepsilon\right)  <1$:%
\[
\left(
\begin{array}
[c]{c}%
E_{t-1}\pi_{t}\\
u_{t}%
\end{array}
\right)  =\left(
\begin{array}
[c]{cc}%
\lambda\left(  q,\varepsilon\right)  & -\frac{1}{\beta q}-\frac{\kappa}{\beta
q}F_{u}\\
0 & \rho
\end{array}
\right)  ^{t}\left(
\begin{array}
[c]{c}%
\frac{\lambda\left(  q,\varepsilon\right)  }{1-\beta q\rho\lambda\left(
q,\varepsilon\right)  }\\
1
\end{array}
\right)  u_{0}.
\]

The positive correlation ($\kappa>0$) between current inflation and current
output gap implies a negative correlation ($-\frac{\kappa}{\beta q}$) between
future inflation and current output gap. The positive sign of the policy rule
parameter ($F_{\pi}\left(  q,\varepsilon\right)  >0$) satisfies this necessary
condition in order to lean against inflation spirals:%

\[
0<\lambda\left(  q,\varepsilon\right)  =\frac{1}{\beta q}-\frac{\kappa}{\beta
q}F_{\pi}\left(  q,\varepsilon\right)  <1<\frac{1}{\beta q}\Rightarrow
-\frac{\kappa}{\beta q}F_{\pi}\left(  q,\varepsilon\right)  <0.
\]

The sign of the correlation of expected inflation with the policy instrument
determines the sign of the response of the policy instrument to the policy
target in the optimal policy rule. By contrast, in the accelerationist
Phillips curve, Clarida, Gali and Gertler [2001] mention that there is a
positive sign of the correlation between expected inflation and current output
gap ($-\frac{\kappa}{\beta q}$). If $\kappa<0$, because negative feedback
requires $-\frac{\kappa}{\beta q}F_{\pi}<0$, the sign of the optimal policy
rule is negative ($F_{\pi}<0$).

\section{DISCRETION}

There is a long history of different definitions of "discretion" for
stabilization policy (Chatelain and Ralf [2020a]). Clarida, Gali and Gertler
[1999] and Gali ([2015], chapter 5) define "discretion" (or discretion
equilibrium) as the case where policy makers re-optimize with certainty each period.

\begin{proposition}
When policy makers re-optimize each period ($q=0$), they do static
optimization each period even if the private sector's transmission mechanism
is dynamic.
\end{proposition}

\begin{proof}
If $q=0$, the policy maker's discounted loss function boils down to a
\emph{static utility} ($\left(  \beta.0\right)  ^{0}=1$, $\left(
\beta.0\right)  ^{t}=0,t=1,2,...$). The policy maker only values the current
period, as he knows he will be replaced next period, whatever the duration of
the next period:%
\[
V^{jj}\left(  u_{0}\right)  =-E_{0}%
{\displaystyle\sum\limits_{t=0}^{t=+\infty}}
\left(  \beta q\right)  ^{t}\left[  \frac{1}{2}\left(  \pi_{t}^{2}%
+\frac{\kappa}{\varepsilon}x_{t}^{2}\right)  +\beta\left(  1-q\right)
V^{jk}\left(  u_{t}\right)  \right]  =-\frac{1}{2}\left(  \pi_{0}^{2}%
+\frac{\kappa}{\varepsilon}x_{0}^{2}\right)  .
\]
The policy maker's infinite horizon transversality condition is always
satisfied as he does not survive the current period and as his discount factor
is zero after his single period of life. Because the policy maker does not
value the future, he does not take into account private sector expectations in
the new-Keynesian Phillips curve. He only considers a static Phillips curve
with an exogenous intercept:%
\[
\pi_{0}=\kappa x_{0}+\beta0E_{0}\pi_{1}+\beta\left(  1-0\right)  E_{0}\pi
_{1}^{k}+u_{0}\text{ }=\kappa x_{0}+u_{0}+\beta E_{0}\pi_{1}^{k=1}\text{,
}u_{0}+\beta E_{0}\pi_{1}^{k+1=1}\text{ given.}%
\]
The superscript for the policy maker index $k$ is now identical to the period
index, as each period corresponds to a new policy maker.
\end{proof}

As expected inflation is not taken into account by successive policy makers,
one order (or one dimension) of the dynamics of inflation is removed in
discretion. By contrast, for a non-zero probability of not reneging
commitment, ($q>0$), the policy maker takes into account private sector
expectations. This causes the bifurcation of the dynamic system between
discretion ($q=0$) versus quasi-commitment ($q\in\left]  0,1\right]  $), see
section 3.

\begin{proposition}
The policy makers' static optimizations of an otherwise dynamic transmission
mechanism implies locally unstable dynamics of the dynamic transmission mechanism.
\end{proposition}

\begin{proof}
On an iso-utility curve, the derivative of the loss function is equal to zero:%
\[
dL=\pi_{0}d\pi_{0}+\frac{\kappa}{\varepsilon}x_{0}dx_{0}=0\Rightarrow
\frac{\partial\pi_{0}}{\partial x_{0}}=-\frac{\kappa}{\varepsilon}\frac{x_{0}%
}{\pi_{0}}.
\]
The policy maker's first order condition is such that the iso-utility ellipse
is tangent to the slope of the static Phillips curve with a given intercept:%
\[
\frac{\partial\pi_{0}}{\partial x_{0}}=\kappa=-\frac{\kappa}{\varepsilon}%
\frac{x_{0}}{\pi_{0}}\Rightarrow x_{0}=-\varepsilon\pi_{0}.
\]
This static optimization implies a proportional policy rule with the exact
negative correlation of the policy instrument (output gap) with the policy
target (inflation). The proportional parameter is the opposite of the
household's elasticity of substitution between goods. Any increase of current
inflation is instantaneously related to a decrease of current output. This
static optimal program is repeatedly solved each period by each new policy
maker of period $t$:%
\begin{equation}
x_{t}=-\varepsilon\pi_{t}\text{ for }t=0,1,2,...\text{ with }\varepsilon
>1\text{.}%
\end{equation}
Substituting this policy rule in the new-Keynesian Phillips curve, one has
this recursive unstable dynamics for inflation, denoted the discretionary
inflation eigenvalue $\lambda\left(  0,\varepsilon\right)  $ for $q=0$:%
\[
E_{t}\pi_{t+1}^{t+1}=\left(  \frac{1+\kappa\varepsilon}{\beta}\right)  \pi
_{t}-\frac{1}{\beta}u_{t}\text{ with }1<\frac{1}{\beta}<\lambda\left(
0,\varepsilon\right)  =\frac{1+\kappa\varepsilon}{\beta}\text{ .}%
\]
The policy maker lives only the current period. He sets a zero weight on
future inflation. For the policy maker, it does not matter what happens as
long as it is not in his instantaneous lifetime.\emph{ The policy maker does
not care in his loss function that his static optimal policy leads to
sub-optimal unstable dynamics in the future.} The economy dynamic system is
locally unstable with the inflation dynamics eigenvalue outside the unit
circle ($\lambda\left(  0,\varepsilon\right)  >\frac{1+\kappa}{\beta}>1$) and
the cost-push shock autoregressive root inside the unit circle ($0<\rho<1$):%
\begin{equation}
\left(
\begin{array}
[c]{cc}%
\lambda\left(  0,\varepsilon\right)  =\frac{1}{\beta}+\frac{\kappa}{\beta
}\varepsilon & -\frac{1}{\beta}\\
0 & \rho
\end{array}
\right)  \left(
\begin{array}
[c]{c}%
\pi_{t}\\
u_{t}%
\end{array}
\right)  =\left(
\begin{array}
[c]{c}%
E_{t}\pi_{t+1}^{t+1}\\
u_{t+1}%
\end{array}
\right)  .
\end{equation}

\end{proof}

The optimal policy for a static Phillips curve assuming a positive correlation
between current inflation and current output gap ($\kappa>0$) but excluding
inflation expectations, leads to an optimal solution where the output gap
responds negatively to inflation ($F_{\pi}=-\varepsilon<0$). This example
shows that the optimal solution of static optimization can be different from
the solution of dynamic optimization (Phillips [1954]).

The policy makers' sub-optimal static optimizations of a dynamic transmission
mechanism implies locally unstable dynamics in the space of policy targets.
Starting with

Hypothesis $H_{1}$: Policy makers re-optimize every period, $q=0$,

building epicycles on epicycles, four additional restrictions on private
sector's behavior have to be imposed in order to restrict inflation dynamics
to the stable subspace of dimension one within the unstable space of dimension
two. These four assumptions are not explicitly spelled out in Clarida, Gali
and Gertler [1999] and in Gali [2015]:

Hypothesis $H_{2}$: The private sector expectation of inflation $E_{t}\pi
_{1}^{k+1=1}$ does not depend on current and past values of inflation.

Hypothesis $H_{3}$: The private sector does not select paths with inflation or
deflation spirals.

Hypothesis $H_{4}$: The policy instrument $x_{t}$ is a non-predetermined
variable without a given initial condition $x_{0}$.

Hypothesis $H_{5}$: Policy makers and the private sector measure with infinite
precision the initial and the future values of variables and the structural
parameters of the transmission mechanism.

\begin{proposition}
Under assumptions $H_{i}$ ($i=1,...,5$), Blanchard and Kahn's [1980]
determinacy condition forces a unique solution with an exact correlation of
inflation and output gap with the cost-push shock.
\end{proposition}

\begin{proof}
Blanchard and Kahn's [1980] solution is given by the unique slope of the
eigenvectors of the given stable eigenvalue $0<\rho<1$ of the predetermined
cost-push shock:%
\[
\left(
\begin{array}
[c]{cc}%
\frac{1}{\beta}+\frac{\kappa}{\beta}\varepsilon & -\frac{1}{\beta}\\
0 & \rho
\end{array}
\right)  \left(
\begin{array}
[c]{c}%
\pi_{t}\\
u_{t}%
\end{array}
\right)  =\rho\left(
\begin{array}
[c]{c}%
\pi_{t}\\
u_{t}%
\end{array}
\right)  \Rightarrow\left(  \frac{1}{\beta}+\frac{\kappa}{\beta}%
\varepsilon-\rho\right)  \pi_{t}=\frac{1}{\beta}u_{t}.
\]
There is an exact positive correlation between inflation and the cost-push
shock:%
\begin{equation}
\pi_{t}=\left(  \frac{1}{1-\beta\rho+\kappa\varepsilon}\right)  u_{t}.
\end{equation}
Combining this equation with the policy rule leads to an exact negative
correlation between output gap $x_{t}$ and the cost-push shock $u_{t}$:%
\begin{equation}
x_{t}=-\varepsilon\left(  \frac{1}{1-\beta\rho+\kappa\varepsilon}\right)
u_{t}\text{.}%
\end{equation}

\end{proof}

The dynamics of inflation is locally unstable in the space ($\pi_{t},u_{t}$)
of dimension two. However, inflation dynamics is constrained to vary in a
stable subspace of dimension one, which is lower than the dimension of the
policy target and forcing variables dynamics.\ For example, if $u_{t}$ is
exogenous imported inflation, inflation $\pi_{t}$ would be exactly correlated
with imported inflation, without any relation to other drivers of inflation
within the country.

For practitioners of monetary policy, the assumption $H_{5}$ of perfect
measurement stretches credulity to its limit, because of the measurement
issues for inflation, for the output gap, for the cost-push shock, for the
slope $\kappa\left(  \beta,\varepsilon,\alpha_{L},\theta,\sigma,\varphi
\right)  $ of the new-Keynesian Phillips curve (Mavroeidis \textit{et al.}
[2015]) and for the correlation of expected inflation with current inflation
($1/\beta$). Because of these imperfect measurements, in the real world, the
probability to select unstable paths with inflation and deflation spirals is
equal to one with the discretion policy rule.

\section{BIFURCATION}

\subsection{Inflation Eigenvalue}

Going from monetary policy with quasi commitment to discretionary policy
results in a fundamental change of the properties of the dynamical system.

\begin{proposition}
For any value of the elasticity of substitution between goods $\varepsilon>1$:

(i) There is a saddle-node bifurcation on the inflation eigenvalue when
shifting from quasi-commitment ($q\in\left]  0,1\right]  $ with stable
eigenvalue $0<\lambda\left(  q,\varepsilon\right)  <1$) to discretion ($q=0$,
with unstable eigenvalue $\lambda\left(  0,\varepsilon\right)  >1$):\emph{ }%
\[
0<\lambda_{\min}<\lambda\left(  q,\varepsilon\right)  <\frac{1}{1+\kappa
_{\max}}<1<\frac{1+\kappa_{\max}}{\beta}<\lambda\left(  0,\varepsilon\right)
<\frac{1}{\beta}+\frac{1}{\beta}\frac{\kappa_{\max}}{\alpha_{L}}\text{.}%
\]

(ii) The optimal inflation persistence decreases with the policy maker's
credibility measured by the probability of not reneging commitment
$q\in\left]  0,1\right]  $:%
\[
\frac{\partial\lambda\left(  q,\varepsilon\right)  }{\partial q}<0\text{.}%
\]

(iii) The inflation eigenvalue decreases (respectively increases) with the
elasticity of substitution $\varepsilon$ for quasi-commitment ($q\in\left]
0,1\right]  $) (respectively for discretion ($q=0$)):%
\[
\frac{\partial\lambda\left(  q,\varepsilon\right)  }{\partial\varepsilon
}<0<\frac{\partial\lambda\left(  0,\varepsilon\right)  }{\partial\varepsilon
}=\frac{\kappa}{\beta}%
\]

\end{proposition}

\begin{proof}
For (i), we first seek the limits of $\kappa\left(  \varepsilon\right)
\varepsilon$ which is an increasing function of $\varepsilon\in\left]
1,+\infty\right[  $,
\begin{align*}
\underset{\varepsilon\rightarrow1^{+}}{\lim}\kappa\left(  \varepsilon\right)
\varepsilon &  =\left(  \sigma+\frac{\varphi+\alpha_{L}}{1-\alpha_{L}}\right)
\frac{\left(  1-\theta\right)  \left(  1-\beta\theta\right)  }{\theta}\left(
1-\alpha_{L}\right)  =\kappa_{\max}\\
\underset{\varepsilon\rightarrow+\infty}{\lim}\kappa\left(  \varepsilon
\right)  \varepsilon &  =\left(  \sigma+\frac{\varphi+\alpha_{L}}{1-\alpha
_{L}}\right)  \frac{\left(  1-\theta\right)  \left(  1-\beta\theta\right)
}{\theta}\frac{\left(  1-\alpha_{L}\right)  }{\alpha_{L}}=\frac{\kappa_{\max}%
}{\alpha_{L}}\text{ with }0<\alpha_{L}<1\text{\ }\\
&  \Rightarrow\kappa_{\max}<\kappa\left(  \varepsilon\right)  \varepsilon
<\frac{\kappa_{\max}}{\alpha_{L}}.
\end{align*}
For discretion, the inflation eigenvalue is an increasing affine function of
$\kappa\varepsilon$ with boundaries:%
\[
1<\frac{1}{\beta}<\frac{1}{\beta}+\frac{1}{\beta}\kappa_{\max}<\lambda\left(
0,\varepsilon\right)  =\frac{1}{\beta}+\frac{1}{\beta}\kappa\varepsilon
<\frac{1}{\beta}+\frac{1}{\beta}\frac{\kappa_{\max}}{\alpha_{L}}.
\]
For quasi-commitment, $\lambda\left(  q,\varepsilon\right)  $ is obtained
solving a linear quadratic regulator model so that the inflation eigenvalue is
necessarily within the range $\left]  0,1\right[  $. It is a decreasing
function of $\beta q$, of $q$ (according to (i)) of $\kappa\varepsilon$ and of
$\varepsilon$ with this upper bound:%
\[
\underset{q\rightarrow0^{+}}{\lim}\underset{\varepsilon\rightarrow1^{+}}{\lim
}\lambda\left(  q,\varepsilon\right)  =\frac{1}{2}\left(  1+\frac{1}{\beta
q}+\frac{\kappa\varepsilon}{\beta q}\right)  -\sqrt{\frac{1}{4}\left(
1+\frac{1}{\beta q}+\frac{\kappa\varepsilon}{\beta q}\right)  ^{2}-\frac
{1}{\beta q}}=\frac{1}{1+\kappa_{\max}}<1
\]
which is true because:%
\[
\underset{q\rightarrow0^{+}}{\lim}\frac{1}{2}\left(  1+\frac{1}{\beta q}%
+\frac{1}{\beta q}\kappa_{\max}\right)  -\sqrt{\frac{1}{4}\left(  1+\frac
{1}{\beta q}+\frac{1}{\beta q}\kappa_{\max}\right)  ^{2}-\frac{1}{\beta q}%
}=\frac{1}{1+\kappa_{\max}}<1
\]
and because when $q\rightarrow0^{+}$:%
\[
\lambda\left(  q,\varepsilon\right)  \sim\frac{1+\kappa}{2\beta q}\left(
1-\sqrt{1-\frac{4\beta q}{\left(  1+\kappa\right)  ^{2}}}\right)  \sim
\frac{1+\kappa}{2\beta q}\frac{1}{2}\frac{4\beta q}{\left(  1+\kappa\right)
^{2}}=\frac{1}{1+\kappa}.
\]
Its lower bound is strictly positive:%
\[
\underset{q\rightarrow1^{-}}{\lim}\underset{\varepsilon\rightarrow
+\infty}{\lim}\lambda\left(  q,\varepsilon\right)  =\lambda_{\min}=\frac{1}%
{2}\left(  1+\frac{1}{\beta}+\frac{1}{\beta}\frac{\kappa_{\max}}{\alpha_{L}%
}\right)  -\sqrt{\frac{1}{4}\left(  1+\frac{1}{\beta}+\frac{1}{\beta}%
\frac{\kappa_{\max}}{\alpha_{L}}\right)  ^{2}-\frac{1}{\beta}}>0.
\]
For (ii), see appendix. For claim (iii):
\begin{align*}
\text{sign}\frac{\partial\lambda\left(  q,\varepsilon\right)  }{\partial
\varepsilon}  &  =\text{sign}\frac{\kappa}{\beta q}-\frac{\frac{1}{2}\left(
1+\frac{1}{\beta q}+\frac{\varepsilon\kappa}{\beta q}\right)  \left(
\frac{\kappa}{\beta q}\right)  }{\sqrt{\left(  1+\frac{1}{\beta q}%
+\frac{\varepsilon\kappa}{\beta q}\right)  ^{2}-\frac{4}{\beta q}}}\\
&  =\text{sign}\sqrt{\left(  1+\frac{1}{\beta q}+\frac{\varepsilon\kappa
}{\beta q}\right)  ^{2}-\frac{4}{\beta q}}-\left(  1+\frac{1}{\beta q}%
+\frac{\varepsilon\kappa}{\beta q}\right) \\
&  =\text{sign}\left(  -2\lambda\left(  q,\varepsilon\right)  \right)  <0
\end{align*}

\end{proof}

\subsection{Rule parameter}

The inflation rule parameter $F_{\pi}$ is an affine and decreasing function of
the inflation eigenvalue $\lambda$:%

\[
F_{\pi}\left(  q,\varepsilon\right)  =\frac{1}{\kappa}-\frac{\beta q}{\kappa
}\lambda\left(  q,\varepsilon\right)  .
\]

\begin{proposition}
For any value of the elasticity of substitution between goods $\varepsilon>1$,
the inflation policy rule parameter $F_{\pi}\left(  q,\varepsilon\right)  $ is
positive and increasing with the elasticity of substitution for
quasi-commitment ($q\in\left]  0,1\right]  $), whereas $F_{\pi}\left(
0,\varepsilon\right)  $ is negative, below $-1$ and decreasing with the
elasticity of substitution $F_{\pi}\left(  0,\varepsilon\right)  $ for
discretion:%
\begin{align*}
-\infty &  <F_{\pi}\left(  0,\varepsilon\right)  =-\varepsilon<-1<0\text{
}<F_{\pi}\left(  q,\varepsilon\right)  =\frac{\lambda\left(  q,\varepsilon
\right)  }{1-\lambda\left(  q,\varepsilon\right)  }\varepsilon.\\
\frac{\partial F_{\pi}\left(  0,\varepsilon\right)  }{\partial\varepsilon}  &
<0<\frac{\partial F_{\pi}\left(  q,\varepsilon\right)  }{\partial\varepsilon}%
\end{align*}

\end{proposition}

\begin{proof}
For quasi-commitment, the policy rule parameter of the response to inflation
is a decreasing function of credibility $q$ and an increasing function of the
elasticity of substitution $\varepsilon$. To prove that the policy rule is
positive, it is sufficient to prove:%
\[
\underset{q\rightarrow1^{-}}{\lim}\underset{\varepsilon\rightarrow1^{+}}{\lim
}\frac{1}{\kappa}-\frac{\beta q}{\kappa}\left(  \frac{1}{2}\left(  1+\frac
{1}{\beta q}+\frac{\kappa\varepsilon}{\beta q}\right)  -\sqrt{\frac{1}%
{4}\left(  1+\frac{1}{\beta q}+\frac{\kappa\varepsilon}{\beta q}\right)
^{2}-\frac{1}{\beta q}}\right)  >0.
\]
When $q=1$ and when $\varepsilon\rightarrow1^{+}$%
\[
F_{\pi}\left(  q,\varepsilon\right)  >\frac{1}{\kappa_{\max}}-\frac{\beta
}{\kappa_{\max}}\left(  \frac{1}{2}\left(  1+\frac{1}{\beta}+\frac
{\kappa_{\max}}{\beta}\right)  -\sqrt{\frac{1}{4}\left(  1+\frac{1}{\beta
}+\frac{\kappa_{\max}}{\beta}\right)  ^{2}-\frac{1}{\beta}}\right)  >0.
\]
In this case, one shows in the appendix that $F_{\pi}>0$ is equivalent to
$\kappa_{\max}+\beta>\beta$ which is true because $\kappa_{\max}>0$.
\end{proof}

\section{IMPULSE\ RESPONSE\ FUNCTIONS\ AND\ WELFARE}

\subsection{Initial anchor of inflation}

To ensure stability after a cost-push shock, inflation has to jump to the
stable manifold. The size of this initial anchor depends on the elasticity of
substitution between goods and the credibility of the policy maker.

\begin{proposition}
For any value of the elasticity of substitution between goods $\varepsilon>1$,

(i) The initial anchor (or jump) of inflation on the cost-push shock decreases
with the elasticity of substitution between goods for both quasi-commitment
and discretion.

(ii) The initial jump of inflation is an increasing function of the limited
credibility $q$ of the policy maker.

(iii) The initial anchor of near-zero credibility is always strictly smaller
than the initial anchor in the case of zero credibility:
\[
\pi_{0}\left(  q,\varepsilon\right)  =\frac{\lambda\left(  q,\varepsilon
\right)  }{1-\beta q\rho\lambda\left(  q,\varepsilon\right)  }u_{0}\text{
}\leq\text{ }\pi_{0}\left(  0,\varepsilon\right)  =\frac{1}{1-\beta\rho
+\kappa\left(  \varepsilon\right)  \varepsilon}u_{0}.\text{ }%
\]

\end{proposition}

\begin{proof}
(i) It is straightforward to check that:
\[
\frac{\partial\pi_{0}\left(  q,\varepsilon\right)  }{\partial\varepsilon
}<0\text{, }\frac{\partial\pi_{0}\left(  0,\varepsilon\right)  }%
{\partial\varepsilon}<0
\]

For discretion, the anchor of inflation is a decreasing function of
$\kappa\left(  \varepsilon\right)  \varepsilon$ which is an increasing
function of $\varepsilon$. As $\kappa_{\max}<\kappa\left(  \varepsilon\right)
\varepsilon<\frac{\kappa_{\max}}{\alpha_{L}}$, the zero credibility initial
anchor of inflation ($\pi_{0}/u_{0}$) is bounded as follows:%
\[
0<\frac{1}{1-\beta\rho+\frac{\kappa_{\max}}{\alpha_{L}}}<\frac{1}{1-\beta
\rho+\kappa\varepsilon}<\underset{\varepsilon\rightarrow1}{\lim}\frac
{1}{1-\beta\rho+\kappa\varepsilon}=\frac{1}{1-\beta\rho+\kappa_{\max}}\text{.}%
\]

For limited credibility, the anchor of inflation is a decreasing function of
$\kappa\left(  \varepsilon\right)  \varepsilon$ which is an increasing
function of $\varepsilon$. As $\kappa_{\max}<\kappa\varepsilon<\frac
{\kappa_{\max}}{\alpha_{L}}$, the non-zero credibility initial anchor of
inflation ($\pi_{0}/u_{0}$) upper bound is:%
\[
\underset{q\rightarrow1^{-}}{\lim}\underset{\varepsilon\rightarrow1}{\lim
}\frac{\lambda}{1-\beta q\rho\lambda}=\underset{q\rightarrow1^{-}}{\lim
}\underset{\varepsilon\rightarrow1}{\lim}\frac{\lambda}{1-\beta\rho\lambda
}>1\text{,}%
\]
with:%
\[
\underset{q\rightarrow1^{-}}{\lim}\underset{\varepsilon\rightarrow1}{\lim
}\lambda\left(  q,\varepsilon\right)  =\frac{1}{2}\left(  1+\frac{1}{\beta
}+\frac{\kappa_{\max}}{\beta}\right)  -\sqrt{\frac{1}{4}\left(  1+\frac
{1}{\beta}+\frac{\kappa_{\max}}{\beta}\right)  ^{2}-\frac{1}{\beta}}<1.
\]
(ii) It is straightforward to check that:
\[
\frac{\partial\pi_{0}\left(  q,\varepsilon\right)  }{\partial q}>0\text{ for
}q\in\left]  0,1\right]
\]
(iii) The initial anchor of near-zero credibility is always strictly smaller
than the initial anchor in the case of zero credibility. The gap tends to zero
when the auto-correlation of the forcing variable tends to zero and when the
elasticity of substitution tends to one: $\rho\rightarrow0$ and $\varepsilon
\rightarrow1$.%
\[
\underset{q\rightarrow0^{+}}{\lim}\frac{\lambda\left(  q,\varepsilon\right)
}{1-\beta q\rho\lambda\left(  q,\varepsilon\right)  }=\underset{q\rightarrow
0^{+}}{\lim}\lambda\left(  q,\varepsilon\right)  \sim\frac{1}{1+\kappa\left(
\varepsilon\right)  }<\frac{1}{1-\beta\rho+\kappa\left(  \varepsilon\right)
\varepsilon}%
\]
For Gali's [2015] calibration ($\rho=0.8$), for any elasticity of substitution
$\varepsilon>1$ and for any probability of not reneging commitment
$q\in\left]  0,1\right]  $, the zero credibility initial anchor of inflation
is much higher than the one with minimal credibility.
\end{proof}

\subsection{Impulse response functions}

The values of the parameters are taken from Gali's [2015] calibration:
$\rho=0.8$, $\beta=0.99$, $\varepsilon=6$, $\kappa=0.1275$ obtained with
$\theta=2/3$, $1-\alpha_{L}=2/3$, $\sigma=1$ and $\varphi=1$. Expected impulse
response functions are shown in figure 1 for four different degrees of
credibility $q$: $0$ (discretion (Gali [2015]), $10^{-7}$ (near-zero
credibility), $0.5$ (limited credibility), $1$ (commitment, Gali [2015]).

\textbf{INSERT\ FIGURE\ 1\ HERE.}

The parameters of the inflation dynamics change marginally between $q=1$ and
$q=10^{-7}$. Inflation eigenvalue increases from $\lambda=0.43$ to $0.57$.
Inflation sensitivity with a lagged cost-push shock shifts from $-0.13$ to
$-0.08$. Inflation initial anchor on a cost-push shock shifts from $0.65$ to
$0.57$.

By contrast, the shifts from near-zero credibility $q=10^{-7}$ to zero
credibility $q=0$ are large. Inflation eigenvalue increases from
$\lambda=0.57$ to $1.78$ (multiplied by $3$, crossing the bifurcation value
$1$). Inflation sensitivity with a lagged cost-push shock shifts from $-0.08$
to $-1.01$ (multiplied by $12$). Inflation initial anchor on a cost-push shock
shifts from $0.57$ to $1.03$ (multiplied by $1.8$).

The impulse response function of inflation of zero credibility is markedly
above the impulse response functions of inflation with limited credibility,
including near-zero credibility.

To check the lack of robustness to misspecification of discretion, we compute
two impulse response functions facing a $\pm10\%$ error of the initial anchor
of inflation. For quasi-commitment with near-zero credibility ($q=10^{-7}$),
the error gap of $10\%$ with respect to the perfect knowledge optimal path at
the initial date is reduced to less than $1\%$ after eight quarters (figure
2). For discretion ($q=0$), the error gap of $10\%$ with respect to the
perfect knowledge optimal path at the initial date is increased to $110\%$
after four quarters due to inflation and deflation spirals in sharp contrast
with near-zero credibilility $q=10^{-7}$ paths with $\pm10\%$ initial error
also represented on figure 3.

\textbf{INSERT\ FIGURE\ 2\ AND\ FIGURE\ 3\ HERE.}

\subsection{Welfare Losses}

We denote the welfare of discretionary policy $W(0)$. It is usually computed
using households' discount factor of $\beta=0.99$ instead of policy maker's
discount factor $\beta q=0$:%

\begin{align*}
W(0)  &  =-\frac{1}{2}\sum_{t=0}^{t=+\infty}\beta^{t}\left(  \pi_{t}^{2}%
+\frac{\kappa}{\varepsilon}x_{t}^{2}\right)  =-\frac{1}{2}\left(
1+\frac{\kappa}{\varepsilon}\varepsilon^{2}\right)  \left(  \frac{1}%
{1+\kappa\varepsilon-\beta\rho}\right)  ^{2}\sum_{t=0}^{t=+\infty}\beta
^{t}\left(  \rho^{t}u_{0}\right)  ^{2}\\
W(0)  &  =-\frac{1}{2}\frac{1+\kappa\varepsilon}{\left(  1+\kappa
\varepsilon-\beta\rho\right)  ^{2}}\frac{u_{0}^{2}}{1-\beta\rho^{2}%
}=-5.\allowbreak09.
\end{align*}

In table 1, for comparison with the welfare of discretionary policy, the
limited credibility welfare is computed using households' discount factor of
$\beta=0.99$ instead of policy maker's discount factor $\beta q$. We simulate
the model over $200$ periods in order to compute welfare for different
elasticities and different levels of credibility. One can also compute welfare
losses of Ramsey optimal policy solving a Riccati equation (Chatelain and Ralf
[2020b]).\bigskip

\textbf{Table 1:} Welfare loss in percentage of welfare loss with infinite
horizon commitment ($w(q)=\frac{W(q)}{W(1)}-1$, $\beta=0.99$) when varying the
elasticity of substitution $\varepsilon$ and credibility $q$%

\begin{tabular}
[c]{|l|l|l|l|l|l|l|l|l|}\hline
- & - & - & $q=1$ & $0.8$ & $0.5$ & $0.1$ & $10^{-7}$ & $0$\\\hline
$\varepsilon$ & $\kappa\left(  \varepsilon\right)  $ & $\frac{\kappa\left(
\varepsilon\right)  }{\varepsilon}$ & $W\left(  1\right)  $ & $w(q)$ & $w(q)$
& $w(q)$ & $w(q)$ & $w(q)$\\\hline
$3193$ & $0.00032$ & $10^{-7}$ & $-2.119$ & $2.8\%$ & $6.8\%$ & $10.8\%$ &
$2.1\%$ & $73\%$\\\hline
$6$ & $0.1275$ & $0.02125$ & $-2.688$ & $3.2\%$ & $7.4\%$ & $10.9\%$ &
$0.03\%$ & $89\%$\\\hline
$2.35$ & $0.235$ & $0.1$ & $-3.489$ & $3.6\%$ & $7.8\%$ & $10.2\%$ & $8.6\%$ &
$111\%$\\\hline
$1.001$ & $0.34$ & $0.34$ & $-7.971$ & $4.1\%$ & $7.8\%$ & $7\%$ & $23.6\%$ &
$141\%$\\\hline
\end{tabular}

\bigskip

Because there is a wide gap between the large impulse response functions of
zero-credibility $q=0$ with respect to near zero credibility $q=10^{-7}$, the
welfare gap between near-zero versus zero credibility is also gigantic: from
$71\%$ if $\varepsilon=3193$ to $117\%$ when $\varepsilon$ tends to one.

When considering only limited credibility cases, the losses with respect to
infinite horizon commitment are at most an increase of $24\%$ of welfare
losses in the limit case of the elasticity of substitution tending to $1$,
(corresponding to a large relative weight on output gap in the loss function
of $0.34$) for all the range of non-zero probabilities of reneging commitment.

\section{CONCLUSION}

When the probability to renege commitment tends to zero, the quasi-commitment
equilibrium is never the limit of the equilibrium under discretion where the
probability to renege commitment is exactly zero. The discretion equilibrium
is obtained by static optimization of the policy maker of an otherwise
dynamics transmission mechanism, which implies inflation or deflation spirals
as soon as the parameters of the economy are not perfectly known.

FollowingOnce it is clear that the discretion equilibrium is not a relevant
theory for stabilization policy, the empirical issue is the measurement of the
sign of the slope of the new-Keynesian Phillips curve. If it turns out to be
negative, the transmission mechanism corresponds to an accelerationist
Phillips curve instead of the new-Keynesian Phillips curve. Unfortunately, for
the U.S., nearly $50\%$ of the estimates have a positive sign in a large
number of estimations done by Mavroeidis et al. [2014]. Once this sign is
known, the optimal response of the policy instrument to inflation will have
the opposite sign under quasi-commitment.

\section{APPENDIX}

The new-Keynesian Phillips curve can be written as a function of the Lagrange
multiplier where $\kappa>0$, $0<\beta<1$ and $0<q<1$:
\[
E_{t}\pi_{t+1}+\frac{\kappa\varepsilon}{\beta q}\gamma_{t+1}=\frac{1}{\beta
q}\pi_{t}-\frac{1}{\beta q}u_{t}-\frac{1-q}{q}E_{t}\pi_{t+1}^{j}\text{ }%
\]

We keep the notation of Gali [2015] chapter 5 of the Lagrange multiplier with
one step ahead subscript $\gamma_{t+1}$.The Hamiltonian system is:%

\[
\left(
\begin{array}
[c]{ccc}%
1 & \frac{\kappa\varepsilon}{\beta q} & 0\\
0 & 1 & 0\\
0 & 0 & 1
\end{array}
\right)  \left(
\begin{array}
[c]{c}%
\pi_{t+1}\\
\gamma_{t+1}\\
u_{t+1}%
\end{array}
\right)  =\left(
\begin{array}
[c]{ccc}%
\frac{1}{\beta q} & 0 & \frac{-1}{\beta q}\\
-1 & 1 & 0\\
0 & 0 & \rho
\end{array}
\right)  \left(
\begin{array}
[c]{c}%
\pi_{t}\\
\gamma_{t}\\
u_{t}%
\end{array}
\right)  +\left(
\begin{array}
[c]{c}%
-\frac{1-q}{q}E_{t}\pi_{t+1}^{j}\\
0\\
0
\end{array}
\right)  .
\]

This leads to:%

\[
\left(
\begin{array}
[c]{c}%
\pi_{t+1}\\
\gamma_{t+1}\\
u_{t+1}%
\end{array}
\right)  =\left(
\begin{array}
[c]{ccc}%
\frac{1}{\beta q}+\frac{\kappa\varepsilon}{\beta q} & -\frac{\kappa
\varepsilon}{\beta q} & -\frac{1}{\beta q}\\
-1 & 1 & 0\\
0 & 0 & \rho
\end{array}
\right)  \left(
\begin{array}
[c]{c}%
\pi_{t}\\
\gamma_{t}\\
u_{t}%
\end{array}
\right)  +\left(
\begin{array}
[c]{c}%
-\frac{1-q}{q}E_{t}\pi_{t+1}^{j}\\
0\\
0
\end{array}
\right)  .
\]

The characteristic polynomial of the upper square matrix of the Hamiltonian
system (with a determinant equal to $\frac{1}{\beta q}$) is:%

\[
\lambda^{2}-\left(  1+\frac{1}{\beta q}+\frac{\kappa\varepsilon}{\beta
q}\right)  \lambda+\frac{1}{\beta q}=0.
\]

The Hamiltonian matrix has two stable roots $\rho$ and $\lambda$ and one
unstable root $\frac{1}{\beta q\lambda}$:
\[
\lambda=\frac{1}{2}\left(  1+\frac{1}{\beta q}+\frac{\kappa\varepsilon}{\beta
q}-\sqrt{\left(  1+\frac{1}{\beta q}+\frac{\kappa\varepsilon}{\beta q}\right)
^{2}-\frac{4}{\beta q}}\right)  <\sqrt{\frac{1}{\beta q}}<\frac{1}{\beta
q\lambda}.
\]

Policy rule parameter function of $\lambda\left(  \varepsilon\right)  $ and
$\varepsilon$:%

\begin{align*}
\left(  1-\lambda\right)  \left(  1-\frac{1}{\beta q\lambda}\right)   &
=-\frac{\kappa\varepsilon}{\beta q}\Longrightarrow\left(  \frac{1-\lambda
}{\beta q\lambda}\right)  \left(  \frac{\beta q\lambda-1}{\kappa}\right)
=-\frac{\varepsilon}{\beta q}\Longrightarrow\\
F_{\pi}  &  =\frac{1-\beta q\lambda}{\kappa}=\left(  \frac{\lambda}{1-\lambda
}\right)  \varepsilon.
\end{align*}

The Hamiltonian system can be written as a function of the stable eigenvalue
$\lambda$, after eliminating $\varepsilon$:%

\[
\left(
\begin{array}
[c]{c}%
\pi_{t+1}\\
\gamma_{t+1}\\
u_{t+1}%
\end{array}
\right)  =\left(
\begin{array}
[c]{ccc}%
\lambda+\frac{1}{\beta q\lambda}-1 & 1+\frac{1}{\beta q}-\lambda-\frac
{1}{\beta q\lambda} & -\frac{1}{\beta q}\\
-1 & 1 & 0\\
0 & 0 & \rho
\end{array}
\right)  \left(
\begin{array}
[c]{c}%
\pi_{t+1}\\
\gamma_{t+1}\\
u_{t+1}%
\end{array}
\right)  .
\]

\begin{proposition}
Rule parameters $P_{u}$ and $P_{z}$ of the response of the Lagrange multiplier
on inflation to exogenous variables:%
\begin{align}
\gamma_{t}  &  =P_{\pi}\pi_{t}+P_{u}u_{t}\text{ }\\
P_{\pi}  &  =\frac{1}{1-\lambda}>0\text{, }P_{u}=\frac{1}{1-\lambda}%
\frac{\frac{1}{\beta q}}{\rho-\frac{1}{\beta q\lambda}}=\frac{1}{1-\lambda
}\frac{\lambda}{\beta q\lambda\rho-1}<0\text{.}%
\end{align}

\end{proposition}

\begin{proof}
The solution stabilizes the state-costate vector for any initial value of
inflation $\pi_{0}$ and of the exogenous variables $u_{0}$ in a stable
subspace of dimension two within a space of dimension three ($\pi_{t}%
,\gamma_{t},u_{t}$) of the Hamiltonian system. We seek a characterization of
the Lagrange multiplier $\gamma_{t}$ of the form:%
\[
\gamma_{t}=P_{\pi}\pi_{t}+P_{u}u_{t}.\text{ }%
\]
To deduce the control law associated with vector $\left(  P_{\pi}%
,P_{u}\right)  $, we substitute it into the Hamiltonian system:%
\begin{align*}
&  \left(
\begin{array}
[c]{c}%
\pi_{t+1}\\
P_{\pi}\pi_{t+1}+P_{u}u_{t+1}\\
u_{t+1}%
\end{array}
\right) \\
&  =\left(
\begin{array}
[c]{ccc}%
\frac{1}{\beta q}-\left(  1-\lambda\right)  \left(  1-\frac{1}{\beta q\lambda
}\right)  & \left(  1-\lambda\right)  \left(  1-\frac{1}{\beta q\lambda
}\right)  & -\frac{1}{\beta q}\\
-1 & 1 & 0\\
0 & 0 & \rho
\end{array}
\right)  \left(
\begin{array}
[c]{c}%
\pi_{t}\\
P_{\pi}\pi_{t}+P_{u}u_{t}\\
u_{t}%
\end{array}
\right)  .
\end{align*}
We write the last two equations in this system separately:%
\begin{align*}
P_{\pi}\pi_{t+1}+P_{u}u_{t+1}  &  =\left(  P_{\pi}-1\right)  \pi_{t}%
+P_{u}u_{t}\\
u_{t+1}  &  =\rho u_{t}.
\end{align*}
It follows that:%
\[
\pi_{t+1}=\frac{P_{\pi}-1}{P_{\pi}}\pi_{t}+\frac{\left(  1-\rho\right)  P_{u}%
}{P_{\pi}}u_{t}.
\]
The first equation is such that:%
\[
\pi_{t+1}=\left[  \frac{1}{\beta q}-\left(  1-\lambda\right)  \left(
1-\frac{1}{\beta q\lambda}\right)  \right]  \pi_{t}+\left(  1-\lambda\right)
\left(  1-\frac{1}{\beta q\lambda}\right)  \left(  P_{\pi}\pi_{t}+P_{u}%
u_{t}\right)  -\frac{1}{\beta q}u_{t}.
\]
Factorizing:%
\[
\pi_{t+1}=\left[  \frac{1}{\beta q}-\left(  1-\lambda\right)  \left(
1-\frac{1}{\beta q\lambda}\right)  +\left(  1-\lambda\right)  \left(
1-\frac{1}{\beta q\lambda}\right)  P_{\pi}\right]  \pi_{t}+\left[  \left(
1-\lambda\right)  \left(  1-\frac{1}{\beta q\lambda}\right)  P_{u}-\frac
{1}{\beta q}\right]  u_{t}.
\]
The method of undetermined coefficients implies for the first term:%
\begin{align*}
\frac{P_{\pi}-1}{P_{\pi}}  &  =\frac{1}{\beta q}+\left(  1-\lambda\right)
\left(  1-\frac{1}{\beta q\lambda}\right)  \left(  P_{\pi}-1\right)  ,\\
P_{\pi}  &  =\frac{1}{1-\lambda}.
\end{align*}
For the second term:%
\begin{align*}
\frac{\left(  1-\rho\right)  P_{u}}{P_{\pi}}  &  =\left(  1-\lambda\right)
\left(  1-\frac{1}{\beta q\lambda}\right)  P_{u}-\frac{1}{\beta q}%
\Rightarrow\\
\frac{1}{\beta q}  &  =\left(  1-\frac{1}{\beta q\lambda}-1+\rho\right)
\left(  1-\lambda\right)  P_{u}\Rightarrow\\
P_{u}  &  =\frac{1}{1-\lambda}\frac{\frac{1}{\beta q}}{\rho-\frac{1}{\beta
q\lambda}}\Rightarrow\frac{P_{u}}{P_{\pi}}=\frac{\frac{1}{\beta q}}{\rho
-\frac{1}{\beta q\lambda}}=\frac{-\lambda}{1-\lambda\beta q\rho}.
\end{align*}

\end{proof}

\begin{proposition}
Optimal policy rule parameters are given by:%
\begin{align*}
F_{\pi}  &  =\varepsilon\left(  P_{\pi}-1\right)  =\lambda\varepsilon P_{\pi
}=\varepsilon\frac{\lambda}{1-\lambda}=\frac{1-\beta\lambda}{\kappa}\text{
,}\\
F_{u}  &  =\varepsilon P_{u}=\varepsilon P_{\pi}\frac{\lambda}{\beta
\lambda\rho-1}=\varepsilon\frac{1}{1-\lambda}\frac{\lambda}{\beta\lambda
\rho-1}\text{ ,}\\
\frac{F_{u}}{F_{\pi}}  &  =A=\frac{1}{\lambda}\frac{P_{u}}{P_{\pi}}=\frac
{1}{\beta\lambda\rho-1}=\frac{P_{u}}{P_{\pi}-1}=-1+\beta\rho\frac{P_{u}%
}{P_{\pi}}.
\end{align*}

\end{proposition}

\begin{proof}
The first order condition relates the Lagrange multiplier to the policy
instrument:%
\begin{align*}
x_{t}  &  =\varepsilon\gamma_{t+1}=\varepsilon(\gamma_{t}-\pi_{t})\\
x_{t}  &  =F_{\pi}\pi_{t}+F_{u}u_{t}=\varepsilon(\gamma_{t}-\pi_{t}%
)=\varepsilon(P_{\pi}\pi_{t}+P_{u}u_{t}-\pi_{t})\Rightarrow\\
F_{\pi}  &  =\varepsilon(P_{\pi}-1]\text{, }F_{u}=\varepsilon P_{u}\text{ .}%
\end{align*}

\end{proof}

Furthermore, we can find the following lower bound for the policy rule parameter.%

\begin{align*}
\text{If }1-\frac{1}{2}\left(  \beta q+1+\kappa\right)  +\sqrt{\frac{1}%
{4}\left(  \beta q\right)  ^{2}\left(  1+\frac{1}{\beta q}+\frac{\kappa}{\beta
q}\right)  ^{2}-\beta q}  &  >0\text{ }\Leftrightarrow\sqrt{\frac{1}{4}\left(
1+\beta q+\kappa\right)  ^{2}-\beta q}>\frac{1}{2}(-1+\beta q+\kappa)\\
\frac{1}{4}\left(  1+\beta q+\kappa\right)  ^{2}-\beta q  &  >\frac{1}%
{4}(-1+\beta q+\kappa)^{2}\\
\left(  1+\beta q+\kappa\right)  ^{2}-(-1+\beta q+\kappa)^{2}  &
=\allowbreak4\left(  \kappa+\beta q\right)  >4\beta q\\
\kappa+\beta q  &  >\beta q\text{ which is true.}%
\end{align*}

\textbf{Proof of proposition 4 (ii)}: The larger the discount rate, the lower
the discount factor, the less we weight the present, the larger the speed of
convergence, the lower the inflation persistence eigenvalue, which is the root
of the characteristic polynomial:%

\begin{align*}
\lambda^{2}-\left(  1+\frac{1}{\beta q}+\frac{\kappa\varepsilon}{\beta
q}\right)  \lambda+\frac{1}{\beta q}  &  =0.\\
\lambda^{2}-\left(  1+bP\right)  \lambda+P  &  =0.
\end{align*}

With our notations:%
\begin{align*}
S  &  >1\text{, }b=1+\kappa\varepsilon>1\text{, }P=\frac{1}{\beta q}\geq1\\
\Delta &  >0\Leftrightarrow\left(  1+bP\right)  ^{2}>4P
\end{align*}

Let's prove that the inflation eigenvalue is an increasing function of the
product of the two roots, which is an inverse function of the discount factor
$\beta q$. When proven true, the inflation eigenvalue is a decreasing function
of the discount factor $\beta q$.%

\begin{align*}
sign\left\{  \lambda^{\prime}\left(  P\right)  \right\}   &  =sign\left\{
b-\frac{2b\left(  Pb+1\right)  -4\allowbreak}{2\sqrt{\left(  Pb+1\right)
^{2}-4P}}\allowbreak\right\} \\
&  =sign\left\{  \sqrt{\left(  Pb+1\right)  ^{2}-4P}-\left(  Pb+1\right)
+\frac{2\allowbreak}{b}\right\} \\
&  =sign\left\{  \sqrt{\left(  Pb+1\right)  ^{2}-4P}-\left(  Pb+1\right)
+\frac{2\allowbreak}{b}\right\} \\
&  =sign\left\{  -2\lambda\left(  P\right)  +\frac{2\allowbreak}{b}\right\} \\
\text{ }\lambda^{\prime}\left(  P\right)   &  >0\Leftrightarrow b\lambda
\left(  P\right)  <1\text{ for }b>0\text{ and }\lambda\left(  P\right)  >0
\end{align*}

We use the classic functional inequality:%

\begin{align*}
\sqrt{1+x}  &  <1+\frac{1}{2}x\text{ for }x\geq-1\Rightarrow\\
b\lambda &  =b\frac{1+bP}{2}\left(  1-\sqrt{1-\frac{4P}{\left(  1+bP\right)
^{2}}}\right) \\
&  \leq b\frac{1+bP}{2}\frac{1}{2}\frac{4P}{\left(  1+bP\right)  ^{2}}%
=\frac{bP}{bP+1}<1
\end{align*}

QED.

\end{document}